\documentclass[11pt]{article}
\usepackage{authblk}
\usepackage{graphicx}
\usepackage{subfig}
\usepackage{amsmath,amsfonts,amssymb,amsthm, physics}
\usepackage[margin=1in]{geometry}
\usepackage[colorlinks]{hyperref}
\usepackage{xcolor}
\usepackage{mathtools}

\newcommand{\C}{{\mathbb{C}}}

\newcommand{\R}{{\mathbb{R}}}

\newcommand{\ii}{\mathrm{i}}
\newcommand{\e}{\mathrm{e}}

\theoremstyle{plain}
\newtheorem{theorem}{Theorem}[section]
\newtheorem{lemma}[theorem]{Lemma}

\begin{document}

%%%%%%%%%%%%%%%%%%%%%%%%%%%%%%%%%%%%%%%%%%%%%%%%%%%%%%%%%%%%%%%%%%%%%%%%%%%%%%

\title{Transport properties in directed Quantum Walks on the line}

\author[4]{Rodrigo Chaves\footnote{Corresponding author: rodchaves@ufmg.br}}
\author[3]{Jaime Santos}
\author[1,2]{Bruno Chagas}

\affil[1]{National University of Ireland Galway, Ireland}
\affil[2]{Irish Centre for High-End Computing, Ireland}
\affil[3]{Haslab, INESC TEC, Portugal}
\affil[4]{Universidade Federal de Minas Gerais, Brazil}

\date{\today}
\maketitle

%%%%%%%%%%%%%%%%%%%%%%%%%%%%%%%%%%%%%%%%%%%%%%%%%%%%%%%%%%%%%%%%%%%%%%%%%%%%%%
\begin{abstract}

We obtained analytical expressions considering a directed continuous-time quantum walk on a directed infinite line using Bessel functions, expanding previous results in the literature, for a general initial condition. We derive the equation for the probability distribution, and show how to recover normal and enhanced decay rates for the survival probability by adjusting the phase factor of the direction of the graph. Our result shows that the mean and standard deviation for a specific non-local initial condition does not depend on the direction.
\end{abstract}

\begin{center}
\end{center}

%%%--------- S1: Introduction
\section{Introduction}\label{s:intro}

Random walks have been a very fruitful topic both in physics and computer science. Naturally, a quantum mechanical analog would be defined to investigate its usefulness and to check for discrepancies between the classical and quantum models. Two approaches are possible: continuous-time quantum walks and discrete-time quantum walks; the former will be the focus of this work. Continuous-time quantum walks (CTQW) resulted in numerous applications since it was defined. Farhi and Gutmann \cite{Farhi_1998} created an algorithm that uses CTQW for decision trees. They proved that the algorithm could get polynomial-time solutions for graphs that take exponential time to solve classically. Then, Childs \textit{et al.} \cite{Childs_2002} showed the differences between random walks and CTQW, with the latter having an exponentially faster propagation between nodes than its classical counterpart. The list of applications range from search algorithms \cite{Childs_2004}, universal computation \cite{Childs_2009}, state transfer \cite{Kendon2011, Godsil2012}, transport efficiency \cite{Razzoli2021}, and simulation of many-body systems \cite{lahini2011}.

Directed graphs with complex roots of unity associated with their arcs, also called complex unit gain graphs, have complex Hermitian matrices as adjacency matrices. Recently, there is an increasing interest in them and how those type of directions changes properties in quantum walks. From a graph-theoretical perspective, there are papers associating the combinatorics of arc directions with spectral properties of the complex Hermitian matrices \cite{Mohar2016, Guo2016}. From a dynamics perspective, literature has shown that transport properties for those graphs differ from the undirected case: state transfer for them demands a different characterization \cite{Godsil2020}, new phenomena are possible like zero transfer \cite{Sett2019}, and the dynamics immutability by the addition of direction for certain initial conditions \cite{chaves2022}. 

Another interesting phenomenon in literature \cite{Anderson1958}, and the focus of this work, is localization. In quantum walks, for certain graphs, the walker takes more time to visit all nodes and the time rate for the spreading of the walker is called survival probability. One can induce localization with Hamiltonians perturbations \cite{Candeloro2020}, by adding an oracle in the search problem, in a Bose-Einstein condensate experiment \cite{Delvecchio2020}, or by introducing a spatial inhomogeneity in the coin operators \cite{buarque2019}. Localization is also sensitive to the initial condition of the walk \cite{Abal2006}, where it was observed different spreading rates when the initial condition was changed. However, it also appears for some quantum walk models without disorder \cite{Danac2021}. Our purpose in this work is to analyze localization in an infinite line with directions associated with complex numbers. This will be done analytically by showing a connection between the time evolution of the quantum state and the Bessel functions with general initial conditions.

The paper is organized as follows. In section \ref{s:tb}, we briefly introduce our main tools serving as a theoretical background to the reader: Hermitian oriented graphs, Bessel functions, and CTQW. In section \ref{s:tpa}, we prove the connection between Bessel functions and the time evolution of the quantum state in an infinite line. Also, we analyze the survival probability, moments, and standard deviation of the walker. The paper ends in section \ref{s:c} with the final discussion of the work done and the conclusions drawn from it.

\section{Theoretical Background}\label{s:tb}

In this paper, we will see the definition of the tools that were used to obtain our results. At first, in subsection \ref{ss:og}, we will present Hermitian directed graphs and how it is defined. Following, in subsection \ref{ss:ctqw}, we will see what are CTQW and examples of the dynamics for an infinite line with direction. Then, in section \ref{ss:bf}, we will give a brief description of Bessel functions and some of their properties.

%\textcolor{blue}{
%Nessa parte inicial seria bom dar um panorama da seção.
%Acredito que nenhuma dessas subseções precisa ser muito longa. Pode ser direto e com apenas as definições que precisamos, mas tentando sempre ilustrar com boas figuras e gráficos.
%}
\subsection{Hermitian Directed Graphs}\label{ss:og}

A graph $G$ is a set of nodes $V(G)$ connected by its edge set $E(G)$, and can be represented by a Laplacian matrix from a given graph or by an adjacency matrix $A(G)$. The latter is defined as
\begin{equation}
    A_{ab}=\begin{cases}
    w_{ab}&\text{if $(a,b)\in E(G)$}\\
    0&\text{otherwise}
    \end{cases}.
\end{equation}
When $w_{ab}\in \R$ we say the graph is undirected and the directed case occurs when $w_{ab}\in\C$. The Laplacian matrix is defined similarly, but it has every diagonal entry $A_{aa}$ equal to the degree of the node $a$.

Quantum mechanics demands matrices being Hermitian when considering an isolated system. If there is an arc $(a,b)$ with weight $w_{ab}$, there will be an arc $(b,a)$ with weight $w^*_{ba}$. Additionally, our work will deal with weights that have $\abs{w_{ab}}=1$. Therefore, we can define the Hamiltonian of Hermitian directed graphs as
\begin{equation}\label{eq:alphaadj}
H = \sum_{(a,b) \in E(G)} \e^{\ii \alpha(a,b)}\ketbra{a}{b} + \e^{-\ii \alpha(b,a)}\ketbra{b}{a}, 
\end{equation}
where $\alpha(a,b)\in\R$. Even though $\alpha(a,b)$ can be different for every pair of arcs $(a,b)$ and $(b,a)$, in this work, we use a constant $\alpha$ for all pair of arcs. In this scenario, one can obtain an undirected graph if one sets $\alpha = 0$.  Figure \ref{fig:oriented_line} gives an example of a directed infinite path graph.

Throughout this work, we will study the case of a directed infinite line, which can be defined by the Hamiltonian
\begin{equation}\label{hamiltonian}
    H = \sum_{x = L}^{R}\e^{\ii\alpha}\ketbra{x+1}{x}+\e^{-\ii\alpha}\ketbra{x-1}{x},
\end{equation}
where $L$ and $R$ are the left and right borders of the line where we can obtain an infinite path graph if we set $R\rightarrow\infty$ and $L\rightarrow -\infty$. From now on, since our interest lies on infinite lines and the dynamics associated with it, the reader may assume our Hamiltonian is equal to the one in equation \ref{hamiltonian} unless stated otherwise.

% \begin{figure}[!h]
% \centering
% \includegraphics[width=6cm]{Figuras/oriented_cyle_graph.png}
% \caption{Oriented cycle graph with six vertices}
% \label{fig:oriented_cycle_6}
% \end{figure}

\begin{figure}[!h]
\centering
\includegraphics[width=12cm]{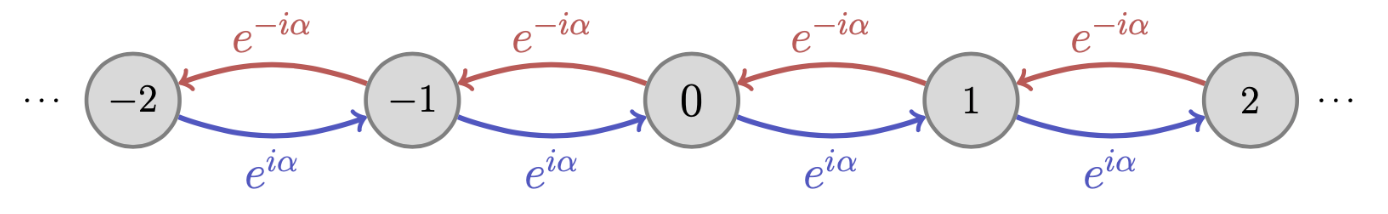}
\caption{Directed infinite line with weights $e^{\ii\alpha}$ and $e^{-\ii\alpha}$.}
\label{fig:oriented_line}
\end{figure}

%\textcolor{blue}{
%\begin{enumerate}
%    \item Definição da hermitean oriented graphs, fazendo algum comentário de como recuperar um grafo simples, alguma conexão com Time-Reversal Symmetry Breaking \cite{Zimbors2013}
%    \item Definir o caso da linha infinita e deixar explícito qual a orientação que vamos usar, e que não é a única possível
%    \item Gráfico da linha infinita orientada
%\end{enumerate}
%}

\subsection{Continuous-Time Quantum Walk}\label{ss:ctqw}

A Hamiltonian $H$ is a self-adjoint operator associated to a Hilbert space. The time evolution of a quantum state, $\ket{\psi(t)}$, in this Hilbert space is governed by Schrödinger's equation
\begin{equation}
    \ii\frac{\partial}{\partial t}\ket{\psi(t)}=H\ket{\psi(t)},
\end{equation}
considering the Planck constant $\hbar=1$. Given an initial condition $\ket{\psi(0)}$ and a time-independent Hamiltonian $H$, the solution to the differential equation will be
\begin{equation}
    \ket{\psi(t)}=\e^{-\ii tH}\ket{\psi(0)}.
\end{equation}
If the Hilbert space is $\C^{2^n}$, where $n$ is the number of particles with two-dimensional states, it is possible to define a $H$, like the Heisenberg Hamiltonian, that has a block decomposition. Each block
is invariant in a $k$-dimensional subspace, where $k$ is the number of $\ket{1}$ in the state. Then, $H$ will have $n+1$ subspaces each with dimension $\binom{n}{k}$ \cite{Osborne2006}. Whenever $k=1$, the block corresponding to this subspace can be described by the adjacency matrix (or the Laplacian matrix depending on the choice of $H$) of the underlying graph that describes the couplings. This description applies for both undirected \cite{Christandl2005} graphs or unitary gain graphs \cite{chaves2022}. 

As an example of the dynamics of the Hamiltonian in equation \ref{hamiltonian}, figure \ref{fig:dynamics} describes the dynamics for an infinite line within a frame of 200 nodes with initial condition
\begin{equation}\label{eq:initialcondition}
    \ket{\psi(0)}=\frac{1}{\sqrt{2}}(\ket{-k}-\ket{k}),
\end{equation}
where $k=3$ in this case and $t=35$. Notice how the value of $\alpha$ changes the dynamics of the graph. While an undirected graph has symmetry around the $0$ node, different values of $\alpha$ can increase the probability of finding the walker on one side or another. This was the main inspiration to seek how these types of weights change the transport properties of the graph. The definition of a quantum state after some time $t$ will be given
\begin{equation}
    \ket{\psi(t)}=\sum_{x= L}^{R}\psi(x,t)\ket{x},
\end{equation}
with $R\rightarrow\infty$ and $L\rightarrow -\infty$. We aim to find an analytical result for the coefficients $\psi(x,t)$ in order to have a full description of the system.

%\textcolor{blue}{
%\begin{enumerate}
%    \item Você pode começar definindo pela equação de schrödinger (eq. 2 em \cite{Razzoli2021}) e que o hamiltoniano é diretamente associado ao laplaciano, mas também com a matriz de adjacência, a qual iremos utilizar
%    \item Depois você defina qual é a solução da equação, mostrando o unitário – lembrando que é para o caso onde o hamiltoniano é independente do tempo
%    \item Gráfico da dinâmica da linha infinita orientada com alguns valores para orientação, e também com condições iniciais distintas
%\end{enumerate}
%}

\begin{figure}[!h]
\centering
\includegraphics[width=14cm]{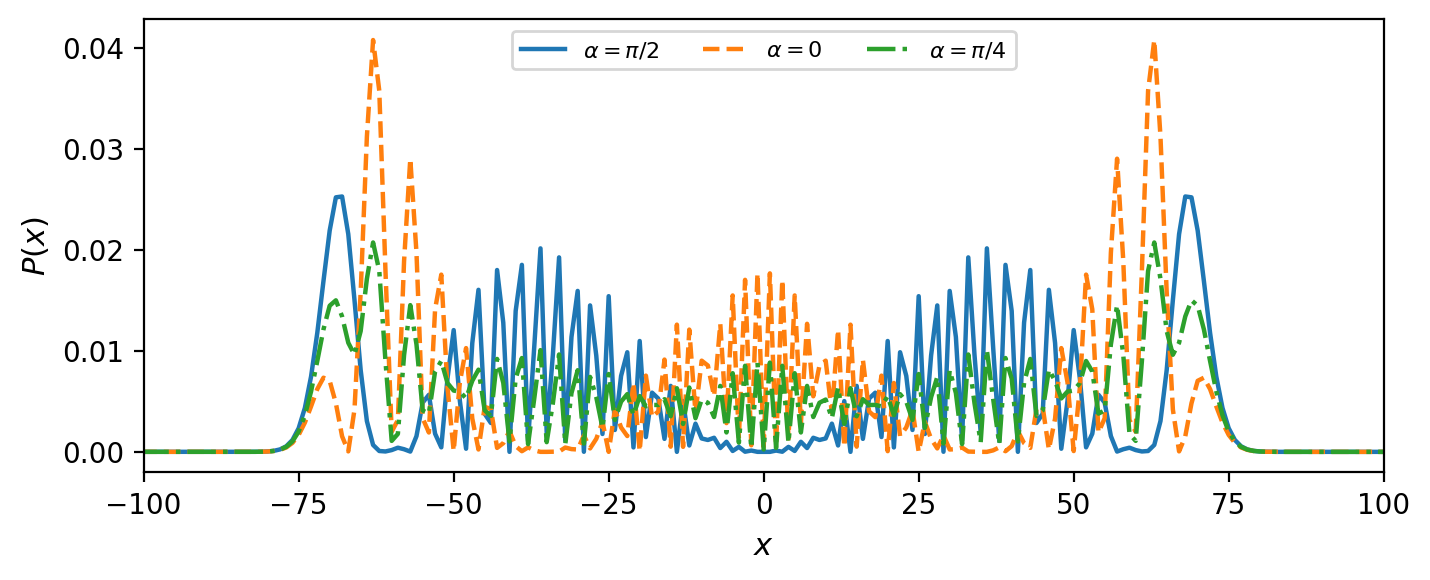}
\caption{Dynamics of directed quantum walk on an infinite line with different angles and $t=35$.}
\label{fig:dynamics}
\end{figure}

\subsection{Bessel Functions}\label{ss:bf}

Bessel functions of the first kind, $J_n(x)$, are solutions to Bessel's differential equations and have numerous representations and relations. They are widely used in physics and the definition most useful for us will be:
\begin{equation}
    J_n(x) = \frac{1}{\pi}\int_{0}^{\pi}\cos{(n\tau - x\sin{\tau})}d\tau = \frac{1}{2\pi}\int_{-\pi}^{\pi} \e^{\ii(n\tau - x\sin\tau)}d\tau.
\end{equation}
The convergence radius is infinity for all Bessel functions \cite{Lebedev1972} and an oscillatory behavior that will dictate the dynamics of the infinite line.

Throughout this paper, we will need some properties from the Bessel equations. The first property shows how to define a function based on these questions, defined as
\begin{equation}\label{eq:tildejn}
    \Tilde{J}_n(x) = \ii^n J_n(x),
\end{equation}
and it will be very present in our demonstrations. We have the recurrence relations
\begin{equation}\label{eq:jnjn1}
     J_{n+1}(x) + J_{n-1}(x)=\frac{2n}{x}J_n(x),
\end{equation}
and
\begin{equation}
J_{n+1}(x) - J_{n-1}(x) = 2 J^{'}_n(x),
\end{equation}
Finally, we will also need 
\begin{equation}\label{eq:relation3}
    J_{-n}(x) = (-1)^n J_n(x).
\end{equation}

% \textcolor{blue}{
% \begin{enumerate}
%     \item Acredito que só falta colocar mais algumas relações que usamos mais pra frente
% \end{enumerate}
% }

\section{Transport Properties Analysis}\label{s:tpa}

This section starts in subsection \ref{subsec:distribution} where the time evolution of the quantum walk on an infinite line is obtained analytically. The result is used to obtain the survival probability in subsection \ref{subsec:survival} and the standard deviation in subsection \ref{subsec:stddeviation}.

%\textcolor{blue}{
%\begin{enumerate}
%    \item Escrever a organização dessa seção. Primeira parte com o cálculo das distribuições. Segunda parte com o decaimento com condição inicial não-localizada, mais o survival probability. Terceira parte o cálculo da média e desvio padrão, caso possível, mas acho qeu dá.
%\end{enumerate}
%}

\subsection{Distributions of Directed Quantum Walk}\label{subsec:distribution}

A better description of the Hamiltonian spectral decomposition will be required to find the time evolution of the CTQW. Hence, the proof starts with the following lemma due to Trench\cite{Trench1985}:
\begin{lemma}\label{lemma:eignVV}
Define a band Toeplitz matrix of dimension $(n-1)\times(n-1)$, $T_k$, with $k<n$,
\begin{equation}
    \begin{bmatrix}
        0&\dots&a_k&&&\\
        \vdots&\ddots&&\ddots&&\\
        a_{-k}&&\ddots&&\ddots&\\
        &\ddots&&\ddots&&a_k\\
        &&&a_{-k}&\dots&0
    \end{bmatrix}.
\end{equation}
If $k=1$ and $a_k=e^{i\alpha}=a^*_{-k}$, where $\alpha\in\mathbb{R}$, then its eigenvalues will be given by
\begin{equation}
    \lambda_k = 2\cos\bigg(\frac{k\pi}{n}\bigg)\qquad k=1,\dots,n-1
\end{equation}
and the corresponding eigenvectors will be
\begin{equation}
    \ket{\psi_k}=\sqrt{\frac{2}{n}}\sum_{x=L+1}^{R-1}e^{-i\alpha x}\sin\bigg(\frac{k\pi x}{n}\bigg)\ket{x}
\end{equation}
\end{lemma}
\qed

The adjacency matrix of a path graph can be described by a Toeplitz matrix like the one above with $k=1$ and, with this knowledge, the main theorem follows
\begin{theorem}
\label{theorem:besselCtqw}
Let $J_n(t)$ denote the n-th Bessel function of the first kind and $\Tilde{J}_n(z)=i^nJ_n(z)$. The continuous-time quantum walk for an infinite line and start point $\ket{\psi(0)}=\ket{x_0}$ has the coefficients:
\begin{equation}
    \psi(x,t) = e^{-i\alpha(x-x_0)}\Tilde{J}_{x-x_0}(2t)
\end{equation}
\end{theorem}
\begin{proof}
    The coefficients $\psi(x,t)$ will be given by
    \begin{equation*}
        \braket{x}{\psi(t)}=e^{itT_1}\braket{x}{\psi(0)}.
    \end{equation*}
    From lemma \ref{lemma:eignVV} we can use the spectral decomposition of $T_1$ to obtain
    \begin{equation*}
        \psi(x,t) = \frac{2}{n}\sum_{k=1}^{n-1}e^{-i\alpha x}\sin\bigg(\frac{k\pi(x-L)}{n}\bigg)e^{it2\cos{(\frac{k\pi}{n})}}\sum_{y=L+1}^{R-1}e^{-i\alpha y}\sin\bigg(\frac{k\pi(y-L)}{n}\bigg)\psi(y,0).
    \end{equation*}
    Set $\psi(y,0) = \delta_{x,x_0}$, then
    \begin{equation*}
        \psi(x,t) = \frac{e^{-i\alpha(x-x_0)}}{n}\bigg(\sum_{k=1}^{n-1}e^{-it2\cos{(\frac{k\pi}{n})}}\cos\bigg(\frac{k\pi(x-x_0)}{n}\bigg)-\sum_{k=1}^{n-1}e^{-it2\cos{(\frac{k\pi}{n})}}\cos\bigg(\frac{k\pi(x+x_0-2L)}{n}\bigg)\bigg)
    \end{equation*}
    Notice that the terms with $k=0$ and $k=n$ vanish, hence we can add them to the sums without loss of generality. After the addition, the sum can be changed to an integral by using Newton-Cotes quadrature. The error of the approximation, $\epsilon(x,x_0,t,n)$, is
    \begin{equation*}
        \epsilon(x,x_0,t,n)\leq\mathcal{O}\bigg(\frac{1}{N^2}\norm{f''(x)}_\infty\bigg),
    \end{equation*}
    where $\norm{f(x)}_\infty$ is the supremum of $f(x)$. Since we are in the interval $[0,1]$, the supremum will always be bounded since the function $f''(x)$ will only have complex exponentials and cosines. The coefficients can now be rewritten as
    \begin{align*}
        \psi(x,t) = & e^{-i\alpha(x-x_0)}\bigg(\int_{0}^{1}\cos{\big(\pi(x-x_0)\phi\big)}e^{i2t\cos{\pi\phi}}d\phi-\int_{0}^{1}\cos{\big(\pi(x+x_0-2L)\phi\big)}e^{i2t\cos{\pi\phi}}d\phi\bigg)+\\&+\epsilon(x,x_0,t,n).
    \end{align*}
    Using \ref{eq:tildejn} we get that
    \begin{equation}
        \psi(x,t) = e^{-\alpha(x-x_0)}\big[\Tilde{J}_{x-x_0}(2t)-\Tilde{J}_{x+x_0-2L}(2t)\big]+\epsilon(x,x_0,t,n).
    \end{equation}
    By setting $R\rightarrow\infty$ and $L\rightarrow -\infty$ we see that the error terms go to zero while one of the Bessel functions goes to zero too. We finally get to our result:
    \begin{equation}
        \psi(x,t) = e^{-i\alpha(x-x_0)}\Tilde{J}_{x-x_0}(2t).
    \end{equation}
\end{proof}

Bessen\cite{Bessen2006} proved a similar result for undirected infinite line graphs where he finds that
\begin{equation}
    \psi(x,t) = \tilde{J}_{\abs{x-x_0}}(2t).
\end{equation}
We can see that the addition of complex roots of unity as weights, introduces relative phases proportional to $\alpha$ between the states associated with each node. 

\subsection{Survival Probability}\label{subsec:survival}
The survival probability of a quantum walk can be characterized as the mean
probability of finding the walker in a certain location after some time
$t$. Considering the symmetric position range of $[k_0,k_1]$, this quantity will
be
\begin{equation}
    P_{[k_0,k_1]}(t)=\sum_{i=k_0}^{k_1} P_{i}(t).
\end{equation}
Choosing initial conditions such that $P_{[k_0,k_1]}(0) = 1$, we can evaluate
the \textit{decay rate} of the survival probability, which indicates how fast the walker leaves the specified range of positions. In a classical random walk the decay rate scales with $t^{-\frac{1}{2}}$, and, typically, with $t^{-1}$ in a quantum walk. However, as shown by \cite{Abal2006}, certain non-local initial conditions can further increase this decay rate, where the survival probability decreases with $t^{-3}$. Here, we will show that this behavior can be recovered for a more general model of continuous-time quantum walks in directed graphs.\par
\begin{figure}[!h]
\centering
\includegraphics[scale=0.35]{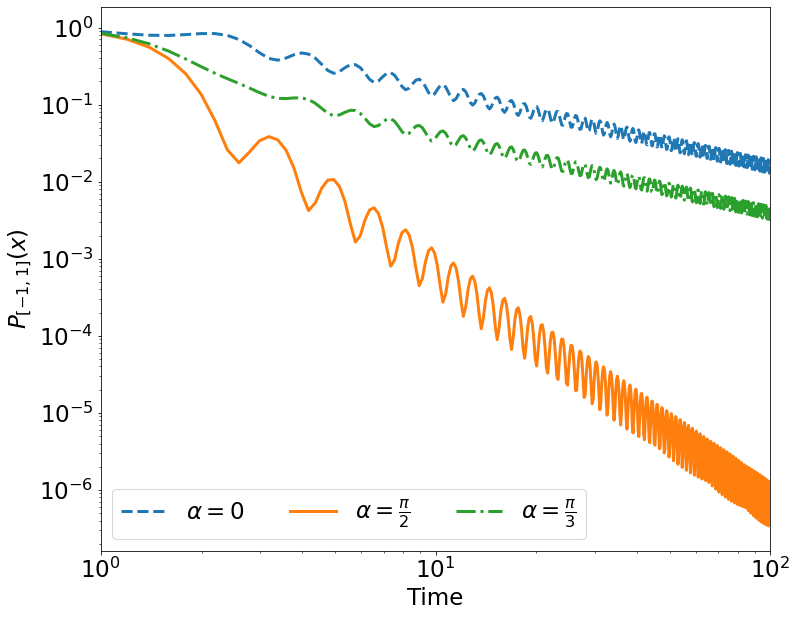}
    \caption{Evolution of the survival probability on a directed infinite line, for $\theta = \frac{\pi}{4}$, $\gamma = 0$, and $\alpha = 0, \frac{\pi}{2},\frac{\pi}{3}$.}
\label{fig:oriented_cycle}
\end{figure}

We define a non-localized initial condition as
\begin{equation}
    \label{eq:survProbInitCond}
    \ket{\psi(0)} = \cos(\theta)\ket{-k} + e^{i\gamma}\sin(\theta)\ket{k},
\end{equation}
where $k \in \mathbb{Z}$. The wave function is then given by 
\begin{equation}
    \psi(x,t) = \cos(\theta) e^{-i\alpha(x+k)}\Tilde{J}_{x+k}(2t)+
    e^{i\gamma}\sin{\theta} e^{-i\alpha(x-k)}\Tilde{J}_{x-k}(2t),
\end{equation}
which leads to the general equation for the probability associated with this initial condition 
\begin{equation}
\begin{aligned}
    P_k(x,t) &= \psi(x,t)\psi(x,t)^{*} \\ 
    &=\cos^2(\theta) J_{x+k}^2(2t) + \sin^2(\theta) J_{x-k}^2(2t)\\
    &\quad+ 2(-1)^k  \cos(2\alpha k +\gamma)  \cos(\theta)\sin(\theta) J_{x+k}(2t)J_{x-k}(2t).
    \end{aligned}
    \label{eq:generalBesselProb}
\end{equation}
Considering a value of $k = 1$ and $\theta = \frac{\pi}{4}$, figure \ref{fig:oriented_cycle} shows how the survival probability decays over time. For a value of $\alpha = \frac{\pi}{2}$ the enhanced decay rate is observed, whereas other values of $\alpha$ display a normal decay rate.\par

The enhanced decay rate can then be analytically obtained from equation \eqref{eq:generalBesselProb} by considering
\begin{equation}
    (-1)^k\cos(2\alpha k + \gamma) = 1,
\end{equation}
and the value of $\alpha$ will be
\begin{equation}
    \label{eq:enhancedAlpha}
    \alpha=\begin{dcases}\frac{2 \pi v - \gamma}{2 k}, \text{even } k, \\ 
                        \frac{\pi+2 \pi v-\gamma}{2 k}, \text { odd } k, 
                        \end{dcases}
\end{equation}
where $v \in \mathbb{Z}$. The probability will be, after some algebraic manipulation, as follows

\begin{equation}
        P_k(x,t) = \left[\cos \theta J_{x+k}(2t) + \sin \theta J_{x-k}(2t)\right]^{2}.
\end{equation}
Considering a value of $k=1$ and $\theta = \frac{\pi}{4}$
\begin{equation}
        P_1(x,t) =  \frac{1}{2}\left[J_{x+1}(2 t)+J_{x-1}(2 t)\right]^2 =2 x^{2}\left[\frac{J_{x}(2t)}{2t}\right]^{2} \sim \frac{1}{t^{3}},
\end{equation}
thus confirming these parameters do indeed lead to the probability of finding the walker in the survival region decreasing with an enhanced rate.\par

Analogously, we can recover the normal decay rate by changing the value of $\alpha$ to
\begin{equation}
    \label{eq:normalAlpha}
    \alpha=\begin{dcases}\frac{2 \pi v - \gamma}{2 k}, \text{odd } k, \\ 
                        \frac{\pi+2 \pi v-\gamma}{2 k}, \text { even } k. 
                        \end{dcases}
\end{equation}
The associated probability will then be 
\begin{equation}
        P_k(x,t) = \left[\cos \theta J_{x+k}(2t) - \sin \theta J_{x-k}(2t)\right]^{2}.
\end{equation}
Once more, if we take a value of $k=1$ and $\theta = \frac{\pi}{4}$
\begin{equation}
P_1(x,t)=\frac{1}{2}\left[J_{x+1}(2t) - J_{x-1}(2t)\right]^2=2\left[\frac{x J_{x}(2t)}{\tau}-J_{x-1}(\tau)\right]^{2} \sim \frac{1}{t},
\end{equation}
and we recover the normal decay rate of the quantum walk.\par

This difference in behavior is explained by an interference effect influenced, in part, by the relative phase of the initial state. In this section, we show that introducing direction to the line graph also affects this interference effect, and we can reproduce enhanced and normal decay rates for the given initial condition by choosing the appropriate value of $\alpha$, regardless of the parity of $k$. Then, no nodes admit localization since the decay rate of right-hand side of for every vertex is O(1/t)

\subsection{Mean and Standard Deviation}\label{subsec:stddeviation}

The mean value, or the first moment, of a distribution for a quantum walk, indicates how the walker moves as a wave function. This transport property can be defined, in our study case, as the equation
\begin{equation}
    \langle x \rangle_t = \sum x\,P_k(x,t),
\end{equation}
where $P_k(x,t)$ is defined by equation \eqref{eq:generalBesselProb} describing the probability distribution depending on $k$. Considering the first moment definition, the considered distribution, and relation \eqref{eq:relation3}, we find the analytical expression for the first moment as
\begin{equation}
    \langle x \rangle_t = -k(\cos^2{\theta} - \sin^2{\theta}) = -k\,\cos{2\theta}.
\end{equation}
This result points out that the first moment does not depend on direction, time, or the phase factor of the initial condition in equation \eqref{eq:initialcondition} – and achieves its maximum value with $\theta= n\pi$ for $n \in \mathbb{Z}$.

Now, we will calculate the standard deviation for this model of the quantum walk, since this transport property is crucial to understanding how the wave spreads in time. For this reason, we introduce the definition of the second moment as
\begin{equation}
    \langle x^2 \rangle_t = \sum x^2\,P(x,t),
\end{equation}
and the standard deviation is written as
\begin{equation}
    \sigma = \sqrt{\langle x^2 \rangle_t -  \langle x \rangle_t^2}.
\end{equation}
Since we have calculated the first moment, we just have to find the second moment. In the same way, we find that
\begin{equation}
    \langle x^2 \rangle_t = 2t^2 + k^2,
\end{equation}
considering $k \neq 1$. In this case, the standard deviation will be written as
\begin{equation}
    \sigma = \sqrt{2t^2 + k^2(1 + \cos^22\theta)} \sim t,
\end{equation}
and slightly depends on $\theta$ and $k$, and presents a ballistic behavior.

\section{Conclusion}\label{s:c}
In theorem \ref{theorem:besselCtqw} we describe the analytical formulation with Bessel functions for the continuous-time quantum walk in a directed line graph. Introducing complex weights to the structure implies relative phases between the coefficients of the wave function, which are still governed by Bessel functions as in the undirected case \cite{Bessen2006} with an extra term to account for graph direction.\par

Our findings expand upon the work of \cite{Abal2006}, where the decay rate depends exclusively on the initial condition. Here, we show how to control the interference patterns by adjusting the value $\alpha$ depending on the parity of $k$ as described by equations \eqref{eq:enhancedAlpha} and \eqref{eq:normalAlpha}. This allows us to recover the normal and enhanced decay rate effects of the survival probability for an arbitrary initial condition. We also find that graph direction does not impact the values of the moments and the standard deviation, they will only depend on the symmetric initial condition of equation \eqref{eq:survProbInitCond}.

In future work, we intend to generalize the initial condition and study how it affects the transport properties of the continuous-time quantum walk. We would also like to find analytical expressions for other classes of graphs, to find if they still exhibit enhanced decay rate, and how would this model apply to other transport properties such as perfect state and transport efficiency.

\section*{Funding}
This study was financed by the Coordenação de Aperfeiçoamento de Pessoal de Nível Superior – Brasil (CAPES) – Finance Code 001.

\bibliographystyle{acm}
\bibliography{biblio}

\clearpage

\end{document}